\documentclass[b5paper,10pt]{article}

\usepackage{mathrsfs}
\usepackage{amssymb}
\usepackage{amsthm} 
\usepackage[shortlabels]{enumitem}
\usepackage{amsmath}
\usepackage{cancel}
\usepackage[all, pdf]{xy}
\usepackage{geometry}
\usepackage{graphicx}
\usepackage{tikz}
\usepackage{pgf}
\usepackage{pgfplots}
\usepackage{comment}
\usepackage{xcolor}
\theoremstyle{definition} 

\sffamily
\newtheorem{thm}{Theorem}[section]

\newtheorem{lem}[thm]{Lemma}
\newtheorem{prop}[thm]{Proposition}
\newtheorem{defn}[thm]{Definition}

\def\E[#1]{\exists^{#1}}
\def\A[#1]{\forall^{#1}}

\title{A Logic that Captures $\beta$P on Ordered Structures}
\author{Kexu WANG,\\ email: wangkexuphy@163.com 
    \and Xishun ZHAO \\ email: hsszxs@mail.sysu.edu.cn\\
Institute of Logic and Cognition,\\
 Department of Philosophy, Sun Yat-sen University}
\date{}

\begin{document}
\maketitle
\thispagestyle{empty}
\begin{abstract}
We extend the inflationary fixed-point logic, IFP, with a new kind of second-order quantifiers which have (poly-)logarithmic bounds. We prove that on ordered structures the new logic $\exists^{\log^{\omega}}\text{IFP}$ captures the limited nondeterminism class $\beta\text{P}$. In order to study its expressive power, we also design a new version of Ehrenfeucht-Fra\"{i}ss\'{e} game for this logic and show that our capturing result will not hold in the general case, i.e., on all the finite structures.
\end{abstract}

\section{Introduction}
In descriptive complexity theory, it is the most interesting task to find a logical characterization of a complexity class. But why do we need logics to characterize (or capture) complexity classes?
\begin{quote}
\emph{Logics speak directly about graphs and structures, whereas most other formalisms operate on encodings of structures by strings or terms. Hence a logical characterization of a complexity class is representation-independent.} 

\raggedleft by Martin Grohe \cite{grohe2011polynomial}
\end{quote}
We know in graph theory or database theory, more essentially we care about \emph{graph properties} (or \emph{Boolean queries}), i.e. the properties which do not depend on encoding. A graph property is always closed under isomorphism. This coincides with that the model class of a logic sentence is closed under isomorphism. Descriptive complexity theory intends to consider every logic sentence as a machine and vice versa. Thus every model of a sentence could be associated with an input of a corresponding machine and the logic (actually a class of sentences) would be related to a complexity class (actually a class of Turing machines). The precise definition will be given in \ref{sec: logic chara}.

In this paper, let's turn to some \emph{limited} (or \emph{bounded}) \emph{nondeterminism} classes, which are included in NP while including P. The idea of limited nondeterminism was first defined by Kintala et al \cite{kintala1984refining}. Then in \cite{cai1997amount} Cai et al discussed a more general case, i.e. the ``Guess-then-Check'' model.
\begin{defn}\label{defn: GC}\cite{cai1997amount}

Let $s: \mathbb N \mapsto \mathbb N$ and $\mathcal C$ be a complexity class. A language $L$ is in the class $GC(s, \mathcal C)$ if there is a language $L^\prime \in \mathcal C$ together with an integer $c>0$ such that for any string $u$, $u \in L$ if and only if $\exists v \in \{0,1\}^*$, $|v|\leq c\cdot s(|u|)$, and $u\# v\in L^\prime$. 
\end{defn}
Naturally NP = $\bigcup_{i \in \mathbb N}GC(n^i, \text{P})$. For any sublinear function $f$, let's define
\[\beta_f = GC(f, \text{P})\]
Specially  for $k\in \mathbb N$ we denote $\beta_k = GC(\log^k, \text{P})$ instead of $\beta_{\log^k}$. Let
\[\beta\text{P}= \bigcup_{k\in \mathbb N} \beta_k\]

Correspondingly we introduce $\exists^{f}$, the second-order quantifier bounded by $f$. (We call this the \emph{f-bounded quantifier}.)  The semantics is straightforward. For any formula $\phi$, any relation variable $X$ and any structure $\mathscr A$,
\begin{align*}
\mathscr A \vDash \exists^f X\phi \Longleftrightarrow &\text{ there is a subset } S\subseteq A^{arity(X)}\text{ with }|S|\leq f(|A|),\\
&\text{ such that } \mathscr A \vDash \phi[\frac{X}{S}]
\end{align*}

We care more about the second-order quantifiers with a logarithmic bound, written as $\exists^{\log^k}$. We call these \emph{log-quantifiers}. The new logic $\exists^{\log^\omega}$IFP is obtained by extending the \emph{inflationary fixed-point logic} IFP with all the log-quantifiers. The main theorem will show that $\exists^{\log^\omega}$IFP captures $\beta$P on ordered structures. An \emph{ordered structure} is a structure whose domain has a built-in linear order. One can notice that the log-quantifiers will act as the part ``$\exists v \in \{0,1\}^*$, $|v|\leq c\cdot s(|u|)$'' in definition \ref{defn: GC}. The log-quantifiers ``guess''  and then the IFP formula will ``check''.

Our characterization is a natural extension of the famous \emph{Fagin's theorem} and \emph{Immerman-Vardi's theorem}. R. Fagin \cite{fagin1974generalized} showed that NP is captured by the existential second-order logic $\Sigma^1_1$, which consists of formulas in the form
\[\exists X_1 \dotsc \exists X_m \phi\]
where $\phi$ is first order and $X_1\dotsc X_m$ are relation variables. As a corollary of Fagin's theorem, every layer of the polynomial time hierarchy, PH, is captured by a layer of the second-order logic \cite{ebbinghaus2005finite}. The fundamental result of capturing P is Immerman-Vardi's theorem \cite{immerman1982relational,vardi1982complexity}. It shows that IFP captures P on ordered structures. 

The restriction on ordered structures is vital. Actually so far we do not know what logic can capture P without a built-in order. Logics are free from encoding, but when we intend to simulate a Turing machine with a logic sentence, it cannot be helped using a linear order to encode graphs or structures. This is related to a more fundamental and sophisticated problem, \emph{canonization} (or \emph{canonical labeling}) of graphs (or structures). A canonization is an algorithm which returns the unique labeling of a graph no matter how we label the vertices of the graph initially. The P-computable canonizations do exist on some certain classes of graphs, for instance, trees \cite{kobler2006graph}, planar graphs \cite{kobler2006graph}, graphs of bounded treewidth \cite{bodlaender1990polynomial}, graphs of bounded degree
\cite{babai1983canonical}. 
Researchers are also interested in using logics to define a canonization. There are IFP-definable canonizations on cycles \cite{ebbinghaus2005finite}, grids \cite{ebbinghaus2005finite} or 3-connected planar graphs \cite{grohe1998fixed}. That means on these classes IFP can provide a canonical linear order and captures P. An important approach is to extend IFP to capture P on some more general classes. For example, IFP with counting, denoted by IFP+$\#$, on trees \cite{immerman1990describing}, planar graphs \cite{grohe1998fixed}, graphs of bounded treewidth \cite{Grohe1999}, graphs of bounded rank width \cite{grohe2019canonisation}. 

Neither IFP nor IFP+$\#$ can capture P in the most general case, i.e., on all the finite structures. They were originally proved via the game method. Alongside this notion we will design a new Ehrenfeucht-Fra\"{i}ss\'{e} game and prove $\exists^{\log^\omega}$IFP fails to capture $\beta$P in the most general case, too.

%
%
%
%

\section{Preliminaries}
We assume that the readers are familiar with the basic concepts of computational complexity theory and mathematical logic. A \emph{signature} $\tau$ is a finite class of relation symbols. For convenience, we do not talk about constant symbols and function symbols. $\mathscr L[\tau]$ is the formulas of logic $\mathscr L$ formed with symbols in $\tau$.  A \emph{$\tau$-structure} (or structure over $\tau$) $\mathscr B$ explains the symbols in $\tau$ on a domain $B$. In this paper we only consider \emph{finite} structures, i.e. whose domain is a finite set. STRUC$[\tau]$ is the class of all $\tau$-structures. A \emph{graph} is a structure over signature $\{E\}$ whose domain $V$ is a set of vertices. STRUC[$\tau$]$_<$ is the class of all ordered $\tau$-structures (there is a built-in linear order of whose domain). STRING is the class of all strings. 
Let $\tau_{\text{str}} = \{<, P_0, P_1, P_\#, P_\langle, P_\rangle \}$.  A string $u$ is a structure over $\tau_{\text{str}}$, i.e. 
\[u= (U, <, P_0^u, P_1^u, P_\#^u, P_\langle^u, P_\rangle^u )\]
where
\begin{itemize}
\item $U=\{0, 1,\dotsc, |u|-1\}$
\item $<$ is the natural linear order of $U$
\item $P_0^ui \Longleftrightarrow \text{the i-th bit of }u\text{ is }0$
\item $P_1^ui \Longleftrightarrow \text{the i-th bit of }u\text{ is }1$ 
\item $P_\#^ui \Longleftrightarrow \text{the i-th bit of }u\text{ is }\#$
\item $P_\langle^ui \Longleftrightarrow \text{the i-th bit of }u\text{ is }\langle$  
\item $P_\rangle^ui \Longleftrightarrow \text{the i-th bit of }u\text{ is }\rangle$  
\end{itemize} 
“\#'' is used to separate two concatenated strings, for instance, ``$u\#v$''. ``$\langle$'' and ``$\rangle$'' are used for encoding in definition \ref{defn:Enumerating encoding}. None of the three auxiliary symbols are theoretically necessary and all strings can be represented binarily, i.e. just with 0 and 1. However their attendance makes our proofs much easier. 

A Boolean query $\mathcal Q$ on $\tau$ is a class of structures over the same signature $\tau$, and closed under isomorphism, i.e. for any $\mathscr A$, $\mathscr B\in$ STRUC[$\tau$], if $\mathscr A \simeq \mathscr B$, then,
\[\mathscr A \in \mathcal Q \Longleftrightarrow \mathscr B \in \mathcal Q\]
For example, \emph{languages} (classes of strings) are Boolean queries on $\tau_{\text{str}}$.

In the following context, we often use the logarithmic function $\log(n)$, whose value is expected to be an integer, so we let $\log(n) = \lceil \log_2(n) \rceil$. Let $[n] = \{0, 1, \dotsc n-1\}$. Note that $\log(n+1)$ is the minimal length of $n$'s binary expression. In this paper, for any formula $\phi(x, X)$, ``$\phi[\frac{x}{a}, \frac{X}{R}]$'' means the value $a$ (resp. $R$) is substituted into $x$ (resp. $X$) if $x$ (resp. $X$) is free. We abuse the notation $|\cdot|$. If $u$ is a string, $|u|$ is its length. If $A$ is a set, $|A|$ is its cardinality. If $\vec x$ is a $k$-tuple, then $|\vec x| = k$.

\subsection{Encoding structures}

In order to represent the structures in a Turing machine, we need to encode structures as strings. W.l.o.g., we take the following way of encoding:
\begin{defn}[Enumerating encoding]
\label{defn:Enumerating encoding}
For any signature $\tau = \{R_1, \dotsc, R_m\}$, where  $arity(R_i) = r_i$ ($1\leq i \leq m$), any $\mathscr A \in$ STRUC$[\tau]_<$ with domain $A = \{a_0, \dotsc, a_{|A|-1}\}$
\begin{enumerate}
\item $enc(\mathscr A) = \langle enc(A)enc(R_1^{\mathscr A})\dotsc enc(R_m^{\mathscr A})\rangle$
\item $enc(A) = \langle enc(a_0)\dotsc enc(a_{|A|-1}) \rangle$
\item For any $i \in \{1, \dotsc, m\}$, suppose $\vec a_1, \dotsc, \vec a_{|R_i^{\mathscr A}|}$ are all $r_i$-tuples in $R_i^{\mathscr A}$, \[enc(R_i^{\mathscr A})  = \langle enc(\vec a_1), \dotsc, enc(\vec a_{|R_i^{\mathscr A}|}) \rangle\]

\item Suppose $\vec t = (t_1, \dotsc, t_s)$ is a tuple with $t_1, \dotsc, t_s \in A$,
\[enc(\vec t) = \langle enc(t_1)\dotsc enc(t_s) \rangle\]

\item Suppose $a$ is the $j$-th element in $A$, $0 \leq j < |A|$,
\[enc(a) = \langle \text{``the } \log|A|\text{-long binary expression of }j \text{''}\rangle\]

\end{enumerate}
Note that 
\[|enc(\mathscr A)| = \log|A| \cdot O(\sum_{1\leq i\leq m}(|R_i^{\mathscr A}|\cdot r_i))\]
and for $i \in \{1, \dotsc, m\}$, 
\[|enc(R_i^{\mathscr A})| = \log|A| \cdot O(|R_i^{\mathscr A}|\cdot r_i)\]

\end{defn}
The length $|enc(\mathscr A)|$ is related to every size $|R_i^{\mathscr A}|$. The machine needs the auxiliary symbols to parse $enc(\mathscr A)$ because it cannot know ahead of time how long $enc(R_i^{\mathscr A})$ is. The extra length of auxiliary symbols can be ignored in a big-Oh notation.

\subsection{Logic characterization of complexity}\label{sec: logic chara}

\begin{defn}
\label{defn:capturing}\cite{grohe2008quest}
A logic $\mathscr{L}$ \emph{captures} a complexity class $\mathcal C$ on a class $\mathcal K$ of structures, if the following conditions are satisfied, 
\begin{enumerate}
\item $\mathscr{L}[\tau]$ is decidable, for any signature $\tau$.
\item There is an effective procedure to associate with each $\mathscr{L}$-sentence $\phi$ a $\mathcal C$-bounded Turing machine M, such that, for any $\mathscr{A} \in \mathcal K$, M can decide whether
\[\mathscr{A} \vDash \phi\]
\item For any Boolean query $\mathcal Q$ in $\mathcal C$, there is an $\mathscr{L}$-sentence $\phi$ such that for any $\mathscr{A} \in \mathcal K$, 
\[\mathscr{A} \vDash \phi \text{  iff  } \mathscr{A} \in \mathcal Q\]
\end{enumerate}
(We assume that $\mathcal K$ is closed under isomorphism.)

If $\mathcal K$ is the class of all structures, we simply say $\mathscr{L}$ captures $\mathcal C$. 
\end{defn}
There are two most classical theorems in descriptive complexity theory.

\begin{thm}[Fagin's Theorem]
\label{thm:Fagin}
\cite{fagin1974generalized}

$\Sigma^1_1$ captures NP. 
\end{thm}
\begin{thm}[Immerman-Vardi Theorem]
\label{thm:Immerman-Vardi}
\cite{immerman1982relational,vardi1982complexity} 

IFP captures P on ordered structures. 
\end{thm}
\noindent IFP is gotten by extending the first-order logic FO with the inflationary fixed-point operator. IFP inherits the formation rules of FO besides
\begin{itemize}
\item If $\psi$ is a formula, then so is $[\text{IFP}_{\vec y\ Y }\psi(\vec y, Y)]\vec t$, where $Y$ is a relation variable and $|\vec y| = |\vec t| = arity(Y)$ 

\end{itemize}
$[\text{IFP}_{\vec y\ Y }\psi(\vec y, Y)]$ is the fixed point of the function $f^{Y \lor \psi}$ defined by the formula $Y\vec y \lor \psi(\vec y, Y)$. This semantics will not be used in this paper, so readers can turn to \cite{ebbinghaus2005finite} and \cite{libkin2013elements} for details.

In logic we needn't even study structures over all different signatures. We particularly care about STRING and graphs, which the structures over other signatures can be interpreted to.

\begin{defn}
Let $\mathscr L$ be a logic. Let $\tau$, $\sigma$ be two signatures. $\sigma = \{R_1, R_2,\dotsc ,R_m\}$, where $arity(R_i)=r_i$ ($1 \leq i \leq m$).  A \emph{$k$-ary $\mathscr L$-interpretation} from $\tau$ to $\sigma$ is a series of $\mathscr L[\tau]$-formulas \[I = (\phi_{\text{uni}}(\vec x), \phi_{R_1}(\vec x_1, \dotsc, \vec x_{r_1}), \dotsc \phi_{R_m}(\vec x_1, \dotsc, \vec x_{r_m}))\]
where the variables $\vec x$, $\vec x_1$, ... are $k$-tuples. For any $\mathscr A$ on $\tau$,
\[I(\mathscr A) = ( \phi^{\mathscr A}_{\text{uni}}(\_), \phi^{\mathscr A}_{R_1}(\_, \dotsc, \_), ..., \phi^{\mathscr A}_{R_m}(\_, \dotsc, \_))\]
is a $\sigma$-structure, if we consider the $k$-tuples satisfying $\phi^{\mathscr A}_{\text{uni}}(\vec x)$ as individual elements.  (Note that: $\phi^{\mathscr A}(\_) := \{\vec a \mid \mathscr A\vDash \phi[\frac{\vec x}{\vec a}]\}$) 

Suppose $S_1\subseteq$ STRUC[$\tau$] and $S_2\subseteq$ STRUC[$\sigma$] are two Boolean queries. If $I$ also makes sure for any $\mathscr A \in$ STRUC[$\tau$], 
\[\mathscr A \in S_1 \Longleftrightarrow I(\mathscr A) \in S_2\]
we say $I$ is an $\mathscr L$-\emph{reduction} from $S_1$ to $S_2$. 
\end{defn}
It is not hard to prove for any $\mathscr A$, $\mathscr B\in$ STRUC[$\tau$]
\[\mathscr A \simeq \mathscr B \Longrightarrow I(\mathscr A) \simeq I(\mathscr B)\]
%

\begin{lem}\label{lem: FO-reduc}
For any signature $\tau$, there is an FO-reduction $I$ from STRUC$[\tau]_<$ to STRING and for any $\mathscr A$, $\mathscr B$ in STRUC$[\tau]_<$,
\[\mathscr A\simeq \mathscr B \Longleftrightarrow I(\mathscr A) = I(\mathscr B)\]
 
\end{lem}

\begin{lem}\label{lem: trans}\cite{ebbinghaus2005finite}
Let $\phi$ be a formula of IFP$[\sigma]$
\[I = \langle \phi_{\text{uni}}, \phi_{R_1}, ..., \phi_{R_m}\rangle\] 
is a $k$-ary reduction from STRUC$[\tau]$ to STRUC$[\sigma]$. There exists an IFP$[\tau]$-sentence $\phi^I$ such that for any $\mathscr A \in$ STRUC$[\tau]$, 
\[\mathscr A \vDash \phi^I \Longleftrightarrow I(\mathscr A) \vDash \phi\]  

%
%
%
%

\end{lem}

These two lemmas tell us STRING and ordered structures are deeply related. In the following context, we will first prove our theorem on STRING, and naturally it holds on ordered structures.

\section{Capturing Results}
Here is an alternative definition of $\beta$P prepared for our later proofs
\begin{defn}
\label{defn:Alternative}
A language $L$ is in the class $\beta_k$ if there is a language $L^\prime \in $ P together with an integer $c>0$ such that for any string $u$, $u \in L$ if and only if $\exists v \in \{0,1\}^{\leq c\cdot \log^k(|u|)}$, $u\# v\in L^\prime$. (where $\{0,1\}^{\leq c\cdot \log^k(|u|)}$ is all the 0-1 strings of length at most $c\cdot \log^k(|u|)$.)

$\beta$P $= \bigcup_{k \in \mathbb N} \beta_k$. 
\end{defn}
Since $\beta_1 = GC(\log, \text{P})$, in fact the ``guess'' part can be computed in time $2^{c\cdot \log}$, which is a polynomial. Thus we have 
\[\text{P} = \beta_1 \subseteq \beta_2 \subseteq \dotsc\subseteq \beta\text{P} \subseteq \text{NP} \]

\subsection{Logarithmic-bounded quantifiers}
The log-quantifier $\exists^{\log^k}$ is the second-order quantifier with a bound $\log^k$. As we mentioned,
\begin{align*}
\mathscr A \vDash \exists^{\log^k} X\phi \Longleftrightarrow &\text{ there is a subset } S\subseteq A^{arity(X)}\text{ with }|S|\leq \log^k(|A|),\\
&\text{ such that } \mathscr A \vDash \phi[\frac{X}{S}]
\end{align*}
It doesn't matter how large $arity(X)$ is. As long as $arity(X)$ is a nonzero natural number, $\exists^{\log^k}$ can be applied. Naturally
\[\forall^{\log^k} X := \neg\exists^{\log^k} X\neg\]
Let $\log^\omega = \{\log^k \mid k > 0\}$. Then $\exists^{\log^{\omega}} = \{\exists^{\log^k} \mid k > 0\}$

\begin{defn}
\label{defn: prenex forms}

An formula of $\exists^{\log^\omega}$IFP is in the form,
\[\exists^{\log^{k_1}} X_1\exists^{\log^{k_2}} X_2\dotsc  \exists^{\log^{k_m}}X_m\psi\]
where $m\geq 0$; $k_1, k_2, \dotsc k_m > 0$; $\psi$ is an IFP-formula.  
\end{defn}
Those formulas without any occurrences of log-quantifiers are \emph{log-quantifier-free}

Here are three parameters we will use. The \emph{maximal variable arity} of a formula, $mva(\phi) = \textbf{max}\{arity(X) \mid X\text{ is a relation variable, free or bounded by a log-quantifier, in }\phi\}$. The \emph{height} of a formula, $height(\phi) = \textbf{max}\{k \mid \exists^{\log^k}\text{ or }\forall^{\log^k}\text{ occurs in }\phi\}$. The \emph{log-quantifier rank} of a formula,
\begin{itemize}
\item $lqr(\phi) = 0$, if $\phi$ is atomic
\item $lqr(\phi) = lqr(\psi)$, if $\phi = \neg\psi$
\item $lqr(\phi) = \textbf{max}(lqr(\psi_1), lqr(\psi_2))$, if $\phi = \psi_1 \to \psi_2$
\item $lqr(\phi) = lqr(\psi)$, if $\phi = \exists x \psi$
\item $lqr(\phi) = lqr(\psi) + 1$, if $\phi = \exists^{\log^k} X \psi$ for $k>0$.

\end{itemize}
For $k>0$, $\exists^{\log^k}$IFP is the sublogic of $\exists^{\log^\omega}$IFP, the heights of whose formulas are no larger than $k$.

\subsection{Main theorem}
\begin{thm}
\label{thm: main theorem}
$\exists^{\log^{\omega}} \text{IFP}$ captures $\beta$P on STRING.
\end{thm}

Actually we will prove for $k \geq 1$, $\exists^{\log^{k}}\text{IFP}$ captures $\beta_{k+1}$ on STRING. Note that an $\exists^{\log^{k}}\text{IFP}[\tau_{\text{str}}]$-sentence corresponds to a $\beta_{k+1}$-bounded Turing machine, not a $\beta_k$-bounded one. It is because for any $u \in$ STRING and any relation variable $X$, when we encode the value of $X$, as we did in definition \ref{defn:Enumerating encoding}, $|enc(X)| = |O(\log^{k+1}|U|)|$. According to definition \ref{defn:capturing}, our proof consists of three parts. The main idea is simple: we use ``$\exists^{\log^{k}}X$'' to simulate ``$\exists v \in \{0,1\}^{\leq c\cdot \log^k(|u|)}$'' in definition \ref{defn:Alternative} and vice versa; then we apply Immerman-Vardi's theorem.

But here is a problem: for any $v$ in ``$\exists v \in \{0,1\}^{\leq c\cdot \log^k(|u|)}$'' in definition \ref{defn:Alternative}, can we have an IFP-reduction $I$ such that there exists $X$ in ``$\exists^{\log^{k}}X$'' and $I(X) = v$?

\begin{lem}
\label{lem:interpretation}
Let $k \in \mathbb N - \{0\}$

There is an encoding $J$ such that for any string $u$ with domain $U$, $J^U$ is a \emph{surjection} from $\{S\mid S\subseteq U^3 \text{ and } |S| \leq \log^k (|U|)\}$ to $\{0, 1\}^{\leq \log^k(|U|)\cdot (\log(|U|) - 1)}$.

And let $\tau_r = \tau_{\text{str}} \cup \{R_1, R_2, \dotsc R_r\}$, where $R_1, \dotsc R_r$ are ternary relation symbols. There is an IFP-reduction $I$ from STRUC[$\tau_r$] to STRING such that for any $u \in $ STRING and binary relations $R_1^u, \dotsc R_r^u \in $ $\{S\mid S\subseteq U^3 \text{ and } |S| \leq \log^k (|U|)\}$, 
\[I(( u, R_1^u, \dotsc R_r^u)) = u\#J^U(R_1^u)J^U(R_2^u)\dotsc J^U(R_r^u)\]
\end{lem}
\begin{proof}
(of lemma \ref{lem:interpretation})

With the help of the linear order $<^u$, we can construct IFP-formula BIT$(y,x)$, which means ``the $x$-th bit of the binary expression of $y$ is 1''.  (But here we do not provide the details of BIT. The readers can check Page 96 of \cite{libkin2013elements}.)

Let $\vec x = x_1x_2x_5x_6yz_1\dotsc z_r x_3x_4$. It's an $(r + 7)$-ary tuple of variables. Now we define:
\begin{align*}
\phi_< \text{ is }&\text{the lexicographic order of these tuples}\\
\phi_{P_0}(\vec x) :=  &(x_1 = x_2 = 0 \land P_0(x_4)) \lor (x_1 = x_2 = 1 \land x_4 = 0)\\
\phi_{P_1}(\vec x) :=  &(x_1 = x_2 = 0 \land P_1(x_4)) \lor (x_1 = x_2 = 1 \land x_4 = 1)\\
\phi_{P_\#}(\vec x) :=  &(x_1 = x_2 = 0 \land P_\#(x_4)) \lor (x_1 = 0 \land x_2 = 1)\\
\phi_{P_{\langle}}(\vec x) := &(x_1 = x_2 = 0 \land P_\langle(x_4))\\
\phi_{P_{\rangle}}(\vec x) := &(x_1 = x_2 = 0 \land P_\rangle(x_4))
\end{align*}
\begin{align*}
\phi_{\text{uni}}(\vec x) := &\text{\phantom{an}} (x_1 = x_2 = 0 \\
						    &\text{\phantom{an}} \land y  = z_1 = \dotsc = z_r = x_3 = x_5 = x_6 = 0)\\
						    &\lor (x_1 = 0 \land x_2 = 1\\
						    &\text{\phantom{an}} \land y = z_1 = \dotsc = z_r = x_3 = x_4 = x_5 = x_6 = 0)\\
						    &\lor (x_1 =  x_2 = 1 \\
						    &\text{\phantom{an}} \land (\bigvee_{1\leq i \leq r}(R_ix_5x_6y \\
						    &\text{\phantom{aaaaaaaaaa}} \land z_i = 0 \land z_1 = \dots = z_{i-1} = z_{i+1} = \dots = z_r = 1\\
						    &\text{\phantom{aaaaaaaaaa}} \land x_6 < \log(|U|) - 1\\
						    &\text{\phantom{aaaaaaaaaa}} \land x_3 \leq x_6\\
						    &\text{\phantom{aaaaaaaaaa}} \land x_4 = 1 \leftrightarrow \text{BIT}(y, x_3)\\
						    &\text{\phantom{aaaaaaaaaa}} \land x_4 = 0 \leftrightarrow \neg \text{BIT}(y, x_3))))
\end{align*}
Let $I = ( \phi_{\text{uni}}, \phi_<, \phi_{P_0}, \phi_{P_1}, \phi_{P_\#}, \phi_{P_{\langle}}, \phi_{P_{\rangle}})$. In each $(r+7)$-ary tuple, $x_1$ and $x_2$ indicates the tuple to be in $u$, $\#$ or $J^U(R_1^u)J^U(R_2^u)\dotsc J^U(R_r^u)$. For any 01 string $v$ with $|v| \leq \log(|u|) - 1$, there exists some $R_i^u$ in which there is a tuple $x_5x_6y$ such that
\begin{itemize}
\item $x_6 + 1 = |v|$ and

\item the first $(x_6 + 1)$ bits of the binary expression of $y$ is exactly $v$.
\end{itemize}
(Note that when we change the bound $\log(|u|) - 1$ to $\log(|u|)$, this will no more hold.) The variable $x_5$ makes sure that $v$ can repeat in $J^U(R_i^u)$ at most $\log(|u|)$ times. This means for any 01 string of length $\leq \log(|u|)\cdot (\log(|u|) - 1)$, we choose a $R_i^u$ to make $J^U(R_i^u)$ equal to it. So $J^U$ is a surjection.
\end{proof}

\begin{proof}
(\textbf{of theorem \ref{thm: main theorem}})

By definition \ref{defn:capturing}, our proof consists of three parts. Let $k >0$.

Firstly. $\exists^{\log^{\omega}}\text{IFP}[\tau]$ is decidable, for any signature $\tau$.

Secondly. For any $\exists^{\log^{k}}\text{IFP}[\tau_{\text{str}}]$-sentence $\phi = \exists^{\log^{k_1}}X_1\dotsc\exists^{\log^{k_m}}X_m \psi$, where $\psi$ is log-quantifier-free and all its relation variables are among $X_1\dotsc X_m$ and $k_1, \dotsc k_m \leq k$. We construct a $\beta_{k+1}$-bounded Turing machine M$_{\phi}$ as follows: for any $u\in$ STRING, 
\begin{align*}
u\vDash \phi \Longleftrightarrow &\text{there are }S_1 \subseteq U^{arity(X_1)}\text{, } \dotsc\text{ , }S_m \subseteq U^{arity(X_m)}\\
&\text{and }|S_1| \leq \log^{k_1}|u|\text{, } \dotsc\text{ , }|S_m| \leq \log^{k_m}|u|\\
&\text{such that }u\vDash \psi[\frac{X_1}{S_1},\dotsc, \frac{X_m}{S_m}]
\end{align*}
By theorem \ref{thm:Immerman-Vardi}, there is a P-bounded Turing machine M$_{\psi}$ that can verify whether
\[\mathscr A \vDash \psi[\frac{X_1}{R_1},\dotsc ,\frac{X_m}{R_m}]\]
for $\mathscr A$ on $\tau_{\text{str}} \cup \{X_1,\dotsc, X_m\}$ and $\mathscr A$'s explanation $R_1,\dotsc,R_m$ of $X_1,\dotsc, X_m$.

In order to guess and store the values of $X_1,\dotsc, X_m$, by definition \ref{defn:Enumerating encoding}, M$_{\phi}$ will need
\[O(\log^{k_1 + 1}|u| \cdot arity(X_1) + \dotsc + \log^{k_m + 1}|u|\cdot arity(X_m))\]
nondeterministic bits, or simply, $O(\log^{k+1}|u|)$ nondeterministic bits in total. 

then M$_{\phi}$ returns TRUE if there are $S_1,\dotsc,S_m$ with $|S_1| \leq \log^{k_1}|u|\text{, } \dotsc\text{ , }|S_m| \leq \log^{k_m}|u|$, such that M$_{\psi}$ accepts $\langle u, S_1,\dotsc,S_m\rangle$. Otherwise M$_{\phi}$ returns FALSE. 

So M$_{\phi}$ is a $\beta_{k+1}$-bounded machine that we want.

Thirdly. Suppose $L$ is a language in $\beta_{k+1}$. By definition \ref{defn:Alternative}, there is a natural number $c$ and a P-bounded Turing machine M, such that for any $u \in$ STRING,
\[u \in L \Longleftrightarrow \exists v \in \{0,1\}^{\leq c\cdot \log^{k+1}(|u|)} \text{ M accepts }u\#v\]
There exists $r \in \mathbb N - \{0\}$ such that for any $n\in \mathbb N - \{0\}$, $c\cdot \log^{k+1}(|u|) \leq r\cdot \log^{k}(n) \cdot (\log(n) - 1)$. Let $R_1, \dotsc R_r$ be $r$ new ternary relation symbol. We can construct a P-bounded Turing machine M$^\prime$ such that for any strings $u, v, z$ with $v\in \{0,1\}^{\leq c\cdot \log^{k+1}(|u|)}$ and $z\in \{0,1\}^{\leq r\cdot \log^{k}(|u|)\cdot(\log(|u|)-1)}$
\begin{align*}
\text{M}^\prime \text{ accepts }u\#z \Longleftrightarrow &\text{M accepts }u\#v\\
&\text{and $v$ is the leftmost }c\cdot \log^{k+1}(|u|)\text{ bits of $z$.}
\end{align*}
By theorem \ref{thm:Immerman-Vardi}, there is an IFP[$\tau_{\text{str}}$]-sentence $\phi_{\text{M}^\prime}$ such that for any $v \in$ STRING, 
\[v \vDash \phi_{\text{M}^\prime} \Longleftrightarrow \text{M}^\prime\text{ accepts }v\]
By lemma \ref{lem:interpretation}, there is a ($r + 7$)-ary IFP reduction from  STRUC$[\tau_{\text{str}}\cup\{R_1, \dotsc R_r\}]$ to STRING, 
$I = \langle \phi_{\text{uni}}, \phi_<, \phi_{P_0}, \phi_{P_1}, \phi_{P_\#}, \phi_{P_\langle},  \phi_{P_\rangle} \rangle$.
With the help of lemma \ref{lem: trans}, let \[\psi := \phi_{\text{M}^\prime}^I\]
$\psi$ is an IFP-sentence on $\tau_{\text{str}}\cup\{R_1, \dotsc R_r\}$. Let
\[\phi = \exists^{\log^k}R_1, \dotsc \exists^{\log^k} R_r \psi\]
which is an $\exists^{\log^{k}}\text{IFP}[\tau_{\text{str}}]$-sentence. And for any $u \in$ STRING,
\[u \in L \Longleftrightarrow u \vDash \phi\]

\end{proof}
\section{The Expressive Power}
IFP fails on a very important P-decidable Boolean query, EVEN \cite{ebbinghaus2005finite}. For any graph $\mathcal G$, 
$\mathcal G \in $ EVEN if and only if domain $|V|$ is even. There is \emph{no} sentence $\phi$ of IFP[$\{E\}$] such that 
\[
\mathcal G \in \text{EVEN} \Longleftrightarrow \mathcal G \vDash \phi
\]
(EVEN is not definable in IFP.) So IFP fails to capture P (on all finite structures). Unfortunately, our strengthened version $\exists^{\log^\omega}$IFP fails, too.

\begin{thm}\label{thm: Failure}
EVEN is not definable in $\exists^{\log^\omega}$IFP.
\end{thm}
IFP's failure was proven via the failure of the infinitary logic $\mathscr L^\omega_{\infty\omega}$. The logic $\mathscr L^s_{\infty\omega}$ is similar to FO, but every formula in $\mathscr L^s_{\infty\omega}$ can have infinite length or infinite quantifier depth but contains at most $s$ variables (free or bounded). Then
\[\mathscr L^\omega_{\infty\omega} = \bigcup_{s \in \mathbb N} \mathscr L^s_{\infty\omega}\]
Readers can turn to Chapter 3 of \cite{ebbinghaus2005finite} for the details. For every single IFP-formula, we can always construct an equivalent $\mathscr L^s_{\infty\omega}$-formula for some $s$. So IFP is a sublogic of $\mathscr L^\omega_{\infty\omega}$. Now we define a new logic $\mathcal L$ (\emph{Beware}! It is not $\mathscr L$!) as follows: for any formula $\phi$, $\psi$ and $\chi$,
\begin{itemize}
\item $\phi \in \mathcal L$ if $\phi \in \mathscr L^\omega_{\infty\omega}$
\item $\exists^{\log^k}X\phi \in \mathcal L$ if $\phi \in \mathcal L$, where $k > 0$ and $X$ is some relation variable. 
\item $\forall^{\log^k}X\phi \in \mathcal L$ if $\phi \in \mathcal L$, where $k > 0$ and $X$ is some relation variable. 
\item $\psi \land \chi \in \mathcal L$ if $\psi \in \mathcal L$ and $\chi \in \mathcal L$
\item $\psi \lor \chi \in \mathcal L$ if $\psi \in \mathcal L$ and $\chi \in \mathcal L$
\end{itemize}
Obviously $\exists^{\log^\omega}$IFP is a sublogic of $\mathcal L$

In order to prove theorem \ref{thm: Failure}, we turn to the game method
\begin{defn}
$\mathscr L$ is any logic. G is a game played by two players, the spoiler and the duplicator, on two structures. we say G is an Ehrenfeucht-Fra\"{i}ss\'{e} game for $\mathscr L$, if for any $\tau$, any $\mathscr A$ and $\mathscr B \in \text{STRUC}[\tau]$, the following are equivalent,
\begin{enumerate}
\item $\mathscr A \equiv^{\mathscr L} \mathscr B$
\item the duplicator wins G($\mathscr A, \mathscr B$)

\end{enumerate}
where ``$\mathscr A \equiv^{\mathscr L} \mathscr B$'' means for any $\mathscr L[\tau]$-sentence $\phi$,
$\mathscr A \vDash \phi$ if and only if $\mathscr B \vDash \phi$.
\end{defn}
For convenience, we use the notation ``$\bar a$'', a lowercase letter with a bar to represent a ordered set of elements and ``$\bar R$'', a capital letter with a bar to represent a ordered set of relations. Please note that $\bar a$ is not tuple $\vec a$. In the following context we will denote $\bar aa = \bar a \cup\{a\}$, $\bar RR = \bar R \cup \{R\}$. If $\vec a$ consists of elements in $\bar a$, we simply say $\vec a$ is \emph{from} $\bar a$. We say $\bar a \mapsto \bar b \in $ Part($\mathscr A, \bar P, \mathscr B, \bar Q$), i.e. $\bar a \mapsto \bar b$ is a \emph{partial isomorphism} from $\langle \mathscr A, \bar R\rangle$ to $\langle \mathscr B, \bar S \rangle$, where $\bar R = \{R_1, \dotsc R_l\}$ and $\bar S = \{S_1, \dotsc S_l\}$, that is, there is a bijection $f$ from $\bar a$ to $\bar b$,
\begin{enumerate}
\item $f(a_i) = b_i$, $a_i \in \bar a$, $b_i \in \bar b$,
\item for any relation $P \in \tau$, and any tuple $\vec t$ from $\bar a$,
\[\vec t \in P^{\mathscr A} \Longleftrightarrow f(\vec t) \in P^{\mathscr B}\]
\item for $1 \leq i \leq l$, and any tuple $\vec t$ from $\bar a$,
\[\vec t\in R_i \Longleftrightarrow f(\vec t)\in S_i\]
\end{enumerate}
In the expansions, actually $\bar R$, $\bar S$ act as new relations.

The Ehrenfeucht-Fra\"{i}ss\'{e} game for $\mathscr L^s_{\infty\omega}$ is the pebble game with $s$ pairs of pebbles, denoted by PG$^s$. In a play of PG$^s(\mathscr A, \mathscr B)$, there are $s$ or less vertices in each of $\mathscr A$ and $\mathscr B$ covered by pebbles. In each move, each player either does nothing, moves one pebble or adds a new pebble (but on each structures there can be at most $s$ pebbles). If the duplicator can make sure the two covered substructures are always isomorphic, then she wins PG$^s(\mathscr A, \mathscr B)$. For the details readers can turn to Chapter 3 of \cite{ebbinghaus2005finite}.

Now let $\mathcal L^{m, r, k, s}$ be the sublogic of $\mathcal L$ such that for any $\phi$ in it,
\begin{itemize}
\item $lqr(\phi)\leq  m$,
\item $mva(\phi)\leq r$,
\item $height(\phi) \leq k$.
\item at most $s$ element variables occur in $\phi$

\end{itemize}Let's design a game $\text{G}^{m, r, k, s}$ for $\mathcal L^{m, r, k, s}$. As $\mathcal L^{m, r, k, s}$ is extended from $\mathscr L^s_{\infty\omega}$ with log-quantifiers in the ``outer layers'', $\text{G}^{m, r, k, s}$ consists of at most $m$ \emph{relation moves} and a game PG$^s$. The players plays a relation move as follows. The spoiler chooses $r^\prime \leq r$ and $k^\prime \leq k$. Then she chooses either $\mathscr A$ or $\mathscr B$. (W.l.o.g. we assume the spoiler chooses $\mathscr A$. Otherwise $\mathscr A$ and $\mathscr B$ are exchanged.) Then she chooses $R \subseteq A^{r^\prime}$ with $|R|\leq \log^{k^\prime}(|A|)$. At last the duplicator chooses $S\subseteq B^{r^\prime}$ with $|S|\leq \log^{k^\prime}(|B|)$.

In a play of G$^{m, r, k, s}(\mathscr A, \mathscr B)$, the spoiler first chooses an arbitrary $m^\prime \leq m$ and they play $m^\prime$ relation moves and then the two structures are expanded as $\langle \mathscr A, R_1, \dotsc R_{m^\prime}\rangle$ and $\langle \mathscr B, S_1, \dotsc S_{m^\prime}\rangle$. Then they play 
$\text{PG}^s(\langle \mathscr A, R_1, \dotsc R_{m^\prime}\rangle, \langle \mathscr B, S_1, \dotsc S_{m^\prime}\rangle)$. Once this pebble game begins, no more relation moves are allowed. If the duplicator wins $\text{PG}^s(\langle \mathscr A, R_1, \dotsc R_{m^\prime}\rangle,$ $  \langle \mathscr B, S_1, \dotsc S_{m^\prime}\rangle)$, she wins the play. 

If she can always win every play, we say she wins (or she has a winning strategy of) G$^{m, r, k, s}(\mathscr A, \mathscr B)$.
\begin{prop}
\label{prop:EF}
For $m \geq 0$, $r, k, s > 0$, $\text{G}^{m, r, k, s}$ is an Ehrenfeucht-Fra\"{i}ss\'{e} game for $\mathcal L^{m, r, k, s}$.
\end{prop}
\begin{proof}
Let $\mathscr A$ and $\mathscr B$ be two structures over a given signature $\tau$.

We construct the isotype of $\mathscr A$. Let $\bar R$ be a set of new relations such that for any $R\in \bar R$, $arity(R)\leq r$ (and $|R| \leq \log^k(|A|)$). 
\[\phi^{0, r, k, s}_{\mathscr A, \bar R}(\bar X) = \bigwedge \{\phi (\bar X) \mid \phi \text{ is an sentence of $\mathscr L^s_{\infty\omega}[\tau \cup \bar X]$ such that }\mathscr A \vDash \phi[\frac{\bar X}{\bar R}]\}\]
then inductively,
\begin{align*}
\phi^{m+1, r, k, s}_{\mathscr A, \bar R}(\bar X) = \bigwedge_{i \leq r}\bigwedge_{j \leq k}&[(\bigwedge_{R \subseteq A^i, |R| \leq log^j(|A|)}\exists^{\log^j}X\phi^{m, r, k, s}_{\mathscr A, \bar RR}(\bar XX))\\
&\land (\forall^{\log^j}X\bigvee_{R \subseteq A^i, |R| \leq log^j(|A|)}\phi^{m, r, k, s}_{\mathscr A, \bar RR}(\bar XX))]
\end{align*}
When $\bar R = \varnothing$, we simply write $\phi^{m, r, k, s}_{\mathscr A}$, which is a sentence of $\mathcal L^{m, r, k, s}$.

Suppose $\mathscr A \equiv^{\mathcal L^{m, r, k, s}} \mathscr B$, then $\mathscr B \vDash \phi^{m, r, k, s}_{\mathscr A}$. The isotype indicates a winning strategy for the duplicator. After $m$ moves if the two structures are expanded as $\langle \mathscr A, \bar R\rangle$ and $\langle \mathscr B, \bar S\rangle$,
\[\mathscr B \vDash \phi^{0, r, k, s}_{\mathscr A, \bar R}[\frac{\bar X}{\bar S}]\]
This means $\langle \mathscr A, \bar R\rangle$ and $\langle \mathscr B, \bar S\rangle$ satisfy the same $\mathscr L^s_{\infty\omega}$-formulas. Therefore the duplicator can win PG$^s(\langle \mathscr A, \bar R\rangle,\langle \mathscr B, \bar S\rangle)$. Then she wins G$^{m, r, k, s}(\mathscr A, \mathscr B)$.

Suppose $\mathscr A \not\equiv^{\mathcal L^{m, r, k, s}} \mathscr B$. There is a sentence $\phi$ of $\mathcal L^{m, r, k, s}$ which $\mathscr A$ and $\mathscr B$ disagree on.  W.l.o.g., we assume that
$\mathscr A \vDash \phi$ and $\mathscr B \nvDash \phi$ and
\[\phi = Q_1^{\log^{k_1}}X_1\dotsc Q_m^{\log^{k_m}}X_m \psi\]
where $\psi$ is an $\mathscr L^s_{\infty\omega}$-sentence and $k_1, \dotsc k_m\leq k$ and $Q_1, \dotsc Q_m \in \{\exists, \forall\}$. Then
\begin{itemize}
\item $\mathscr A \vDash Q_1^{\log^{k_1}}X_1\dotsc Q_m^{\log^{k_m}}X_m \psi$
\item $\mathscr B \vDash \hat Q_1^{\log^{k_1}}X_1\dotsc \hat Q_m^{\log^{k_m}}X_m \neg\psi$

\end{itemize}
(if $Q_i = \exists$, then $\hat Q_i = \forall$; if $Q_i = \forall$, then $\hat Q_i = \exists$, $1\leq i\leq m$.) This provides a winning strategy for the spoiler. In the $i$-th relation move if $Q_i = \exists$ then the spoiler should choose $\mathscr A$ and the relation $R_i \subseteq A^{arity(X_i)}$; otherwise she should choose $\mathscr B$ and the relation $S_i \subseteq B^{arity(X_i)}$. After $m$ relation moves, the structures have been expanded as $\langle\mathscr A, R_1, \dotsc R_m \rangle$ and $\langle\mathscr B, S_1, \dotsc S_m \rangle$.
\begin{itemize}
\item $\langle\mathscr A, R_1, \dotsc R_m \rangle \vDash \psi[\frac{X_1}{R_1}, \dotsc \frac{X_m}{R_m}]$
\item $\langle\mathscr B, S_1, \dotsc S_m \rangle \vDash \neg\psi[\frac{X_1}{S_1}, \dotsc \frac{X_m}{S_m}]$

\end{itemize}
The duplicator cannot win PG$^s(\langle \mathscr A, \bar R\rangle,\langle \mathscr B, \bar S\rangle)$. So she cannot win G$^{m, r, k, s}(\mathscr A, \mathscr B)$. 
\end{proof}

For any $\mathscr A \in $ STRUC[$\tau$], $R\subseteq A^{arity(R)}$ and $a \in A$, we say $R$ \emph{mentions} $a$ (or $a$ is \emph{mentioned} by $R$), if $a$ is a component of some tuple $\vec t \in R$. Let $ment(R) = \{a\in A \mid a \text{ is mentioned by }R\}$. Observe that if $R$ is bounded by log-quantifier $\exists^{\log^k}$,  then
\[|ment(R)| \leq arity(R)\cdot \log^k(|A|)\]
and we denote $ment(\bar R) = \bigcup_{R\in \bar R} ment(R)$

\begin{thm}\label{thm: EVEN FAILURE}
EVEN is not definable in $\mathcal L$.
\end{thm}

\begin{proof}
If EVEN is defined by a sentence $\phi$ of $\mathcal L[\{E\}]$, $\phi$ should also work on empty graphs, namely on the graphs that have no edges. Now we assume $E = \varnothing$ in order to get a contradiction. There are $m \geq 0$ and $r, k, s>0$ such that $\phi \in \mathcal L^{m, r, k, s}[\{E\}]$. Let $\mathscr A$ and $\mathscr B$ be two empty graphs such that $|A|$ is a sufficiently large even number satisfying
\[(m+1)\cdot r\cdot s \cdot \log^k(|A|) < |A|\] 
and $|B| = |A| + 1$ and $\log(|A|) = \log(|B|)$. So $\mathscr A\vDash \phi$ and $\mathscr B\nvDash \phi$. The duplicator can play $\text{G}^{m, r, k, s}(\mathscr A, \mathscr B)$ as follows:

Before this play begins, vacuously $\varnothing \mapsto \varnothing \in\text{Part}(\mathscr A, \mathscr B)$. Let $f: \varnothing \mapsto \varnothing$. Suppose after $i$ moves ($0\leq i < m$), the players have $\langle\mathscr A, \bar R\rangle$ and $\langle\mathscr B, \bar S\rangle$ and $f$ has been extended as $ment(\bar R) \mapsto ment(\bar S)$. In the $(i+1)$-th move, w.l.o.g. the spoiler chooses $R \subseteq A^{r_{i+1}}$ and $|R|\leq \log^{k_{i+1}}(|A|)$, where $r_{i+1} \leq r$ and $k_{i+1}\leq k$. For $a \in ment(R) - ment(\bar R)$, the duplicator can casually choose $b \notin ment(\bar S)$ and extend $f$ with $f(a) = b$. Since 
\[|ment(\bar S)| \leq m\cdot r \cdot \log^k(|B|)\]
which is much smaller than $|B|$, there are enough ``unmentioned'' $b$'s to choose to make $f$ a partial isomorphism. Let
\[S = f(R) = \{(f(t_1), f(t_2), \dotsc f(t_{r_{i+1}}))\mid (t_1, t_2, \dotsc t_{r_{i+1}})\in R\}\]
So the duplicator chooses $S$. the structures are expanded as $\langle \mathscr A, \bar RR\rangle$ and $\langle \mathscr B, \bar SS\rangle$

After $m$ moves, $\mathscr A$ and $\mathscr B$ are expanded as $\langle \mathscr A, R_1, \dotsc R_m \rangle$ and $\langle \mathscr B, S_1, \dotsc S_m \rangle$ which we still denote by $\langle\mathscr A, \bar R\rangle$ and $\langle\mathscr B, \bar S\rangle$ for short. Consider the substructures
\[\langle ment(\bar R), \bar R\rangle \simeq \langle ment(\bar S), \bar S\rangle\]
The other elements which aren't in the substructures are all isolated nodes. One can easily check that the duplicator wins $\text{PG}^s(\mathscr A, \bar R, \mathscr B, \bar S)$.

So the duplicator wins $\text{G}^{m, r, k, s}(\mathscr A, \mathscr B)$. By proposition \ref{prop:EF}, $\mathscr A \equiv^{\mathcal L^{m, r, k, s}} \mathscr B$. That is a contradiction.

So EVEN is not definable in $\mathcal L$. 
\end{proof}
Since $\exists^{\log^\omega}$IFP is a sublogic of $\mathcal L$, EVEN is not definable in $\exists^{\log^\omega}$IFP, either. Hence $\exists^{\log^\omega}$IFP does not capture $\beta$P (on all finite structures).

\section{Furthur Discussion}
Readers might have noticed that the results can be extended onto other complexity classes. For example the existential and universal log-quantifiers can alternate several times in the formula so as to capture a corresponding \emph{limited alternation class}. Furthermore, not only log-quantifiers, we can also consider other second-order quantifier with a bound of cardinality. Let $f$ be a sublinear function on $\mathbb N$. One can easily prove on ordered structures a logic “$\exists^f$IFP” can capture $\beta_{(f\cdot \log)}$, i.e. the complexity class $GC(f(n)\cdot\log(n), \text{P})$, where the parameter ``$\log(n)$'' seems unavoidable. However none of the above can capture the corresponding complexity classes without a linear order. The proofs could be analogous to our theorem \ref{thm: EVEN FAILURE}.

We are not sure 
\begin{itemize}
\item on what \emph{natural} class of graphs, $\exists^{\log^\omega}$IFP can capture $\beta$P while IFP cannot capture P;

\item whether there is a problem in P which $\exists^{\log^\omega}$IFP can define while IFP cannot.
\end{itemize}
These questions could be interesting.
\bibliography{CLN}
\end{document}